\documentclass[runningheads]{llncs}
\usepackage{algorithm}
\usepackage[noend]{algpseudocode}
\usepackage{amsmath}
\usepackage{amssymb}
\usepackage{color}
\usepackage{comment}
\usepackage{enumitem} 
\usepackage{graphicx}
\usepackage{stmaryrd} 
\usepackage[obeyFinal]{todonotes}
\usepackage{float}
\usepackage{appendix}
\usepackage{booktabs}
\usepackage{colortbl}
\usepackage{pifont}
\usepackage{subcaption}

\usepackage{refcount}

\makeatletter
\renewcommand{\@Opargbegintheorem}[4]{%
  #4\trivlist\item[\hskip\labelsep{#3#2\@thmcounterend}]}
\makeatother

\newif\ifnewexample
\newexampletrue

\newcommand{\sQuad}{\hspace*{0.5em}} 

\newcommand{\qedhere}{\qed} 

\newcommand{\true}{\mathit{true}}
\newcommand{\false}{\mathit{false}}

\newcommand{\E}{\mathcal{E}}
\newcommand{\Zplus}{\mathcal{Z}^+}
\newcommand{\Zminus}{\mathcal{Z}^-}
\newcommand{\Zaplus}{Z_{a}^+}
\newcommand{\Zaminus}{Z_{a}^-}
\newcommand{\Zalphaplus}{Z_{\alpha}^+}
\newcommand{\Zalphaminus}{Z_{\alpha}^-}
\newcommand{\Zbplus}{Z_{b}^+}
\newcommand{\Zbminus}{Z_{b}^-}
\newcommand{\Zcplus}{Z_{c}^+}
\newcommand{\Zcminus}{Z_{c}^-}
\newcommand{\pbesEquations}{\E_L}
\newcommand{\pbesEc}{{\pbesEquations}\E_{\Zplus} \E_{\Zminus}}
\newcommand{\pbescore}{\mathit{core}(\E)}

\newcommand{\pbes}{\mathcal{E}}

\newcommand{\bnd}{\mathrm{bnd}}
\newcommand{\occ}{\mathrm{occ}}
\newcommand{\rank}{\mathrm{rank}}

\newcommand{\sig}{\mathrm{sig}}
\newcommand{\D}{\mathbb{D}}
\newcommand{\rhs}{\mathbf{RHS}}
\newcommand{\rhsc}{\rhs^c}
\newcommand{\Ec}{\mathbf{E}^c}

\newcommand{\predvars}{\mathcal{X}}
\newcommand{\act}{\mathcal{A}}
\newcommand{\match}{\mathrm{match}}

\newcommand{\PBES}{\mathrm{PBES}}

\newcommand{\PG}{\mathrm{PG}}
\newcommand{\pgG}{\mathrm{G}}

\newcommand{\solve}[1]{\mathit{solve}(#1)}

\newcommand{\truepbes}{\true(\pbes)}

\newcommand{\combine}[1]{\mathit{combine}(#1)}

\newcommand{\cmark}{\ding{51}}%
\newcommand{\xmark}{\ding{55}}%

\newcommand{\mubox}[1]{[ #1 ]}
\newcommand{\mudiamond}[1]{\langle #1 \rangle}
\newcommand{\sem}[1]{\llbracket #1 \rrbracket}

\usepackage{array}
\newcolumntype{L}{>{$}l<{$}}
\newcolumntype{C}{>{$}c<{$}}
\newcolumntype{R}{>{$}r<{$}}
{\begin{trivlist}\small
\item\begin{tabular}{@{}>{\bfseries}p{2.2em}L@{\ }L@{\ }L@{\ }L@{\ }L@{\ }L@{\ }L@{\ }L}}%
{\end{tabular}\end{trivlist}}

\newenvironment{mcrl2equations}%
{\begin{trivlist}\footnotesize
\item\centering\begin{tabular}{@{\ }L@{\ }L@{\ }L@{\ }L@{\ }L@{\ }L@{\ }L@{\ }L}}%
{\end{tabular}\end{trivlist}}

\usepackage{tikz}
\usetikzlibrary{positioning}
\usetikzlibrary{arrows,automata}
\tikzstyle{every state}=[draw,shape=circle,inner sep=1pt,minimum size=8pt]
\tikzset{initial text={},auto}

\begin{document}
\title{Efficient Evidence Generation \\ for Modal $\mu$-Calculus Model Checking \\ (extended version)}
\titlerunning{Efficient Evidence Generation for Modal $\mu$-Calculus Model Checking}

\author{Anna Stramaglia \and Jeroen J. A. Keiren \and Maurice Laveaux \and Tim A. C. Willemse} 
\authorrunning{A. Stramaglia, J.\,J.\,A. Keiren, M. Laveaux and T.\,A.\,C. Willemse}

\institute{Department of Mathematics and Computer Science, Eindhoven University of Technology, Eindhoven, The Netherlands\\
\email{\{a.stramaglia, j.j.a.keiren, m.laveaux, t.a.c.willemse\}@tue.nl}}

\maketitle

\begin{abstract}
Model checking is a technique to automatically assess whether a model of the behaviour of a system meets its requirements.
Evidence explaining why the behaviour does (not) meet its requirements is essential for the user to understand the model checking result.
Willemse and Wesselink showed that parameterised Boolean equation systems (PBESs), an intermediate format for $\mu$-calculus model checking, can be extended with information to generate such evidence.
Solving the resulting PBES is much slower than solving one without additional information, and sometimes even impossible.
In this paper we develop a two-step approach to solving a PBES with additional information: we first solve its \emph{core} and subsequently use the information obtained in this step to solve the PBES with additional information. We prove the correctness of our approach and we have implemented it, demonstrating that it efficiently generates evidence using both explicit and symbolic solving techniques.
\end{abstract}

\section{Introduction} 
Model checking~\cite{baier_principles_2008,clarke_model_2018} is an automated technique for establishing whether user-defined properties hold for (a model of) a system.
The behaviour of the system is typically specified using a modelling language whose semantics is represented in terms of a \textit{labelled transition system} or a \textit{Kripke structure}.
Requirements are expressed as formulas in LTL (linear temporal logic), or branching-time logics such as CTL (computation tree logic), CTL$^*$, or the modal \textit{$\mu$-calculus}.

Given the description of the system and a temporal logic formula, a \textit{model checker} answers the decision problem: `Does (the model of) my system meet its requirement?'.
The \emph{yes} / \emph{no}  answer alone does not explain why the requirement is (not) satisfied.
To this end, model checkers can provide evidence (often referred to as a witness or a counterexample) explaining the answer.

Model checking tools such as CADP~\cite{garavel_cadp_2013} and mCRL2~\cite{BGK+2019} use \textit{parameterised Boolean equation systems} (PBESs) to encode the \textit{$\mu$-calculus} model checking problem~\cite{groote_model-checking_2005}.
In mCRL2, PBESs are first instantiated to a parity game (or Boolean equation system)~\cite{van_dam_instantiation_2008,kant_efficient_2012,ploeger_verification_2011} using a process similar to state space exploration.
The resulting parity game is solved using standard algorithms such as the recursive algorithm~\cite{zielonka_infinite_1998}.
Wesselink and Willemse showed that the encoding of the model checking problem to PBES can be extended with additional information such that evidence explaining the solution can be extracted~\cite{wesselink_evidence_2018}.
The evidence subsequently allows for constructing a subgraph of the original state space that gives a minimal explanation of the outcome of the verification.

A fundamental problem in model checking is the \textit{state-space explosion problem}: the size of the state space underlying a system model grows exponentially in the number of (parallel) components and state variables.
Symbolic model checking~\cite{burch_symbolic_1992,mcmillan_symbolic_1993}
addresses this problem by using symbolic representations such as binary decision diagrams to compactly store the state space.
These ideas have been extended to symbolically explore and solve the parity game underlying a PBES~\cite{blom_symbolic_2008,BGK+2019,kant_efficient_2012,laveaux_--fly_2022}.
Symbolic PBES solvers are routinely used to solve the $\mu$-calculus model checking problem for large models.
For instance, the Workload Management System (WMS) model described in~\cite{remenska_using_2013} and the Mechanical Lung Ventilator (MLV) model from~\cite{dortmont-MLV-2024} could only be verified using symbolic algorithm.
However, in practice, the running time of solving PBESs with evidence information is so high that waiting for a solution is not an option.\vspace{-0.5em}

\paragraph{Contributions.}
Our main contribution in this paper is a new approach for evidence generation from PBESs. Our approach first solves a PBES without additional information. As a second step, the solution of this PBES is used to simplify the solving of the PBES that does have additional information needed for evidence generation.
We establish the correctness of the approach.

We have implemented this approach in the explicit PBES solver in mCRL2~\cite{BGK+2019}, and added a hybrid approach, in which the first step is performed symbolically. This solution is then used to inform the explicit PBES solver in the second step.

We experimentally demonstrate the effectiveness of our new approaches.
In particular, our experiments show that when the first step is done using the explicit solver, the performance is comparable with the original approach in~\cite{wesselink_evidence_2018}.
When using the symbolic solver for the first step, our approach is able to efficiently generate evidence, also in the cases where this was not feasible before.\vspace{-0.5em}

\paragraph{Related work.}
For a comprehensive overview of diagnostics for model checking, we refer to Busard's thesis~\cite{busard_symbolic_2017}.
We limit ourselves to the closest approaches providing evidence or diagnostics for the modal $\mu$-calculus model checking problem.
Such diagnostics have for instance been described using tableaux~\cite{kick-95} and as two-player games~\cite{stirling_local_1991}.
There are several graph-based approaches describing evidence in the literature. Mateescu~\cite{mateescu_efficient_2000} describes evidence for the alternation free $\mu$-calculus as a subgraph of an extended Boolean graph. Cranen et al.~\cite{cranen_proof_2013} describe \textit{proof graphs}, that are an extension of support sets~\cite{tan_evidence-based_2002}.

Symbolic solving of PBESs and parity games was studied in the context of LTSmin~\cite{kant_efficient_2012} and mCRL2~\cite{laveaux_--fly_2022}.
Symbolic model checking with evidence generation has been implemented for (Probabilistic) CTL in NuSMV~\cite{cimatti_nusmv_2002} and PRISM~\cite{kwiatkowska_prism_2011}.\vspace{-0.5em}

\paragraph{Outline.}Sect.~\ref{sec:preliminaries} introduces the necessary background about the $\mu$-calculus, PBESs, evidence generation, and a running example. In Sect.~\ref{sec:improve-evidence-extraction} we introduce a new approach to generate evidence from PBESs and prove its correctness. We evaluate the approach in Sect.~\ref{sec:implementation-evaluation} and conclude in Sect.~\ref{sec:conclusion}.

\section{Preliminaries} \label{sec:preliminaries}
Our work is embedded in the context in which abstract data types are used to describe and reason about data, and we distinguish their syntax and semantics.
We write data sorts with letters $D,E, \ldots$ and the semantics counterpart with $\D, \mathbb{E}, \ldots.$
We require the presence of Booleans and natural numbers along with their usual operators.
We use $B$ to denote Booleans and $N$ to denote natural numbers $\{0,1,2,3, \dots\}$, with semantic counterparts $\mathbb{B}= \{\true,\false\}$ and $\mathbb{N}$ respectively.
For both sorts we use their semantics operation such as $\wedge$ and $+$ also for the syntactic counterparts. We use $\approx$ to syntactically represent equality.
Furthermore, we have a set $\mathcal{D}$ of data variables $d, d_1, \ldots.$
If a term is open we use the \textit{data environment} $\delta$ that maps each variable in $\mathcal{D}$ to a value of the proper semantic domain.
Given a term $t$, the interpretation function, under the context of a data environment $\delta$, is denoted as $\sem{t} \delta$ which is evaluated in the standard way.
We write $\delta[v/d]$ to denote that value $v$ has been assigned to variable $d$, i.e., $\delta[v/d](d') = v$ if $d' = d$, and $\delta[v/d](d') = \delta(d')$ otherwise. 
We assume that every value $v \in \mathbb{D}$ can be represented by a closed term. With a slight abuse of notation, we also use $v$ syntactically for this closed term.
\subsection{Processes}

In this paper, the behaviour of systems is modelled using \emph{linear process equations} (LPEs)~\cite{groote_mousavi_2014}. 
An LPE consists of a single process definition, parameterised with data, and condition-action-effect rules that may refer to local variables. 

\ifnewexample
\begin{definition}\label{def:LPE}
A linear process equation is an equation of the following form:
\[L(d \colon D)
	 = + \{\sum_{e_\alpha \colon E_\alpha} c_\alpha(d,e_\alpha) \rightarrow \alpha \cdot L(g_\alpha(d,e_\alpha)) \mid \alpha \in \act\}\]
where $+$ denotes a non-deterministic choice among the rules, $d \colon D$ is the state, $\alpha \in \act$ is an action label to which we associate local variable $e_\alpha$ of sort $E_\alpha$; $c_\alpha(d,e_\alpha)$ is a condition, and term $g_\alpha(d, e_\alpha)$ describes the next state.
\end{definition}
An LPE represents the (non-deterministic) choice to perform action $\alpha \in \act$ from a state represented by $d$, if condition $c_\alpha(d,e_\alpha)$ evaluates to $\true$ for some value $e_\alpha$, which when executed updates the state to $g_\alpha(d,e_\alpha)$. 
\fi

We typically write $L$ instead of $L(d \colon D)$ when referring to an LPE and omit the sum when there is no local variable.
The semantics of an LPE, with a closed term $e$ as initial state, is a labelled transition system (LTS) denoted by $L(e)$.

\ifnewexample
\begin{example}\label{exa:running-example}
As a running example we consider a system whose behaviour is modelled by LPE $L$ (left), where $a,b,c \in \act$, and its associated LTS (right), assuming the constant $M \approx 3$:

\noindent
\parbox[b]{0.65\textwidth}{
\begin{mcrl2equations}
    L(s: N) &=&\sum_{n: N} (s \approx 1 \wedge 0 < n < M) \to  a . L(s+n)\\
            &+& \sum_{n:N} (0 < n < s < M) \to b . L(s-n)\\
            &+& (s \approx M) \to c . L(s)
\end{mcrl2equations}
}
\parbox[b]{0.35\textwidth}{
    \begin{center}
    \begin{tikzpicture}[node distance=30pt]
        \node [initial above,initial distance=10pt,state] (1) {$1$};
        \node [state, left= of 1] (2) {$2$};
        \node [state, right= of 1] (3) {$3$};
        \draw[->] 
        (3) edge[loop right] node[right] {$c$} (3)
        (1) edge node {$a$} (3)
        (1) edge[bend left] node[below] {$a$} (2)
        (2) edge[bend left] node[above] {$b$} (1)
;
    \end{tikzpicture}
    \end{center}\qedhere
    }
\end{example}
\fi

\subsection{Modal $\mu$-calculus}
In this paper we consider requirements expressed in the  modal $\mu$-calculus~\cite{kozen_results_1983}.

\ifnewexample
\begin{definition}\label{def:mu-calculus-formulae}
    A $\mu$-calculus formula $\varphi$ is defined by the following grammar:
    \[
    \varphi  ::= b \mid Y \mid \varphi_1 \wedge \varphi_2 \mid \varphi_1 \vee \varphi_2 \mid [\alpha]\varphi \mid  \langle \alpha \rangle \varphi \mid \sigma X . \varphi 
    \]
    where $b$ is a Boolean constant, $X, Y \in \mathcal{F}$ are fixpoint variables of some countable set $\mathcal{F}$, $\sigma \in \{\mu, \nu \}$ is a fixpoint, and $\alpha \in \act$ is an action.
\end{definition}
\fi

\ifnewexample
We only consider formulas that are closed. That is, formulas in which no fixpoint variable $Y$ occurs outside the scope of its binder. For instance, $\mu X. \mubox{a}X$ is allowed, but $\mu X. \mubox{a}Y$ is not allowed.
\fi
For the denotational semantics of the $\mu$-calculus we refer to the literature; see for instance~\cite{kozen_results_1983}.
\ifnewexample
\begin{example}\label{exa:running-formula}
Consider the following $\mu$-calculus formula:
\[
  \mu V . (\mudiamond{a} V \vee \mudiamond{b} V \vee \nu W . \mudiamond{c} W).
\]
This formula expresses that there is a finite path of $a$ and $b$ actions that ultimately ends with an infinite sequence of $c$-transitions. Intuitively, this formula holds for our running example: by executing action $a$ to state $3$ and subsequently executing the self-loop in state $3$, such a path is produced.\qedhere
\end{example}
\fi

\subsection{Parameterised Boolean Equation Systems}

Parameterised Boolean equation systems (PBESs) are systems of fixpoint equations parameterised with data, where the right-hand side is a predicate formula.

\begin{definition}\label{def:pred_formula}\label{def:pbes}
    \emph{Parameterised Boolean equation systems (PBESs)} $\E$ and predicate formulas $\varphi$ are syntactically defined as follows:
    \begin{align*}
        \E &::= \epsilon \mid (\sigma X(d \colon D_\predvars) = \varphi)\,\E \\
        \varphi & ::= b \mid X(e) \mid \varphi_1 \wedge \varphi_2 \mid \varphi_1 \vee \varphi_2 \mid \exists_{d: D} . \varphi \mid \forall_{d: D} . \varphi
    \end{align*}
    where $\epsilon$ is the empty PBES, $\sigma \in \{\mu, \nu\}$ is a fixpoint, $X \in \predvars$ are predicate variables, $d$ are data variables, and $b$ and $e$ are terms over data variables, where $b$ is of sort $B$.    
\end{definition}
For equation $\sigma X(d \colon D_\predvars) = \varphi$, we write $d_X$ and $\varphi_X$ to denote parameter $d$ and predicate formula $\varphi$, respectively.
The set $\bnd(\E)$ is the set of \textit{bound} predicate variables occurring at the left-hand side of the equations in $\E$.
We denote the set of predicate variables \textit{occurring} in formula $\varphi$ with $\occ(\varphi)$, and the predicate variables occurring in the right-hand sides of $\E$ with $\occ(\E)$.
A PBES $\pbes$ is \emph{well-formed} if it has exactly one defining equation for each $X \in \bnd(\pbes)$. It is \emph{closed} if for every $X$, $\occ(\varphi_X) \subseteq \bnd(\pbes)$, and the only free data variable in $\varphi_X$ is $d_X$.

\ifnewexample
\begin{example}\label{exa:core-pbes-running}
    Following the encoding of~\cite{groote_model-checking_2005}, the following PBES encodes whether $L(1)$ of Example~\ref{exa:running-example} satisfies the $\mu$-calculus formula of Example~\ref{exa:running-formula}:
    \begin{mcrl2equations}
     (\,   \mu X(s \colon N) &=& \exists_{n: N}. (s \approx 1 \wedge 0 < n < 3 \wedge X(s+n) ) \\ 
                         &\vee& \exists_{n: N}. (0 < n < s < 3 \wedge X(s-n) ) \\
                         &\vee& Y(s) \,)\\
     (\,   \nu Y(s \colon N) &=& s \approx 3 \wedge Y(s)\,) \\
    \end{mcrl2equations}\qedhere
\end{example}
\fi
Predicate formulas are interpreted in the context of a predicate environment $\eta$ and data environment $\delta$; see Table~\ref{tab:sem-pred-formula} for details.

\begin{table}[!b]
\caption{The interpretation function $\sem{\varphi}\eta \delta$ of predicate formula $\varphi$ is its truth assignment in the context of $\delta$ and $\eta \colon \predvars \to 2^\D$, data and predicate environments.} \label{tab:sem-pred-formula}
    \centering
    \begin{tabular}{l}
        \hline
        \begin{tabular}{l l p{0.1cm} l l}
         $\sem{b}\eta \delta$&= $\sem{b} \delta$ & &
         $\sem{X(e)} \eta \delta$ &= $\eta(X)(\sem{e} \delta)$\\
         $\sem{\varphi \wedge \psi} \eta \delta$ &= $\sem{\varphi} \eta \delta$ \text{ and } $\sem{\psi} \eta \delta$ &  &
         $\sem{\varphi \vee \psi} \eta \delta$ &= $\sem{\varphi} \eta \delta$ \text{ or } $\sem{\psi} \eta \delta$\\
         $\sem{\exists_{d: D}. \varphi} \eta \delta$ &= \text{ for some } $v \in \D, \sem{\varphi} \eta \delta[v/d]$ & & 
         $\sem{\forall_{d: D}. \varphi} \eta \delta$ &= \text{ for all } $v \in \D, \sem{\varphi} \eta \delta[v/d]$\\
         \end{tabular}\\
         \hline
     \end{tabular}
\end{table}

We define the semantics of PBESs using \emph{proof graphs}.
Given PBES $\E$, $\sig(\E) = \{(X,v) \mid X \in \bnd(\E), v\in \D\}$ denotes the signature of $\E$, where $v \in \D$ is a value taken from the domain underlying the type of $X$. Every predicate variable $X \in \bnd(\E)$ is assigned a \textit{rank};  $\rank_{\E}(X)$ is \textit{even} if and only if $X$ is labelled with a greatest fixpoint, and $\rank_{\E}(X) \leq \rank_{\E}(Y)$ if $X$ occurs before $Y$ in $\E$.

\begin{definition}[\cite{cranen_proof_2013}]\label{def:proof_graph}
    Let $\E$ be a PBES and $\pgG = (V,E)$ be a directed graph, where $V \subseteq \sig(\E)$ and $E \subseteq V \times V$. 
    The graph $\pgG$ is a \textit{proof graph} iff:
    \begin{itemize}
        \item for every $X(v) \in V$ and $\delta$, $\sem{\varphi_X} \eta_{X(v)} \delta[v/d_X]= \true$ with $\eta_{X(v)}(Y)(w) = \true$ iff $\langle X(v), Y(w) \rangle \in E$ for all $Y$; 
        \item for all infinite paths $X_1(v_1) X_2(v_2) \ldots$ through $\pgG$, $\mathrm{min}\{\rank_{\E}(X) \mid X \in V^{\infty}\}$ is even, where $V^{\infty}$ is the set of predicate variables that occur infinitely often in the sequence.
    \end{itemize}    
\end{definition}
The first condition states that if all successors of $X(v) \in V$ in $\pgG =(V,E)$ together yield an environment that makes $\varphi_X$ $\true$ when parameter $d_X$ is assigned value $v$, then $X(v) = \true$. 
The second condition ensures that the graph respects the parity condition typically associated with nested fixpoint formulas.
The semantics of PBES $\E$ is now defined as follows~\cite{cranen_proof_2013}.

\begin{definition}\label{def:solution_PGsem}
The semantics of PBES $\E$ is a predicate environment $\sem {\E}$ such that $\sem{\E} (X)(v)$ is $\true$ iff $X {\in} \bnd(\E)$ and $X(v) {\in} V$ for some proof graph $\pgG {=} (V,E)$.
\end{definition}
We use $\PG(\pbes)$ to refer to a proof graph for PBES $\pbes$.
An explanation of $X(v) = \false$ is given by means of a \textit{refutation graph}, the dual of a proof graph, see~\cite{cranen_proof_2013}. Because of their duality we here outline our theory using proof graphs only.

PBESs are commonly solved by using a process akin to state space exploration to obtain a parity game (or Boolean equation system)~\cite{van_dam_instantiation_2008,kant_efficient_2012,ploeger_verification_2011}, and solving the resulting game.
In practice, this process uses syntactic simplifications to reduce the number of vertices that is generated in the parity game.
A conservative estimate of the number of vertices that need to be explored instead relies on \emph{semantic} dependencies. These are captured by \emph{relevancy graphs}~\cite{liemextraction2023}. A relevancy graph contains a dependency $X(v) \to Y(w)$ if changing the truth value of $Y(w)$ can change the truth value of $\varphi_X[d_X := v]$.
In the definition we write $\eta[ b / X(e) ]$ for the predicate environment satisfying $\eta[ b / X(e) ](Y)(f) = b$ if $X = Y$ and $e = f$, and $\eta[ b / X(e) ](Y)(f) = \eta(Y)(f)$ otherwise. 

\begin{definition}[\cite{liemextraction2023}]\label{def:relevancy-graph} 
    Let $\E$ be a $\PBES$ and $\mathit{RG} = (V, \to)$ be a directed graph, where 
    \begin{itemize}
        \item $V \subseteq \sig(\E)$ is a set of vertices,
        \item $\to \subseteq V \times V$ an edge relation such that for any $X(v) \in V$,
         
        \qquad $X(v) \to Y(w)$ iff

        \qquad $\exists \eta, \delta . \sem{\varphi_X}\eta[\true/ Y(w)]\delta[v / d_X] \neq \sem{\varphi_X}\eta[\false / Y(w)]\delta[v/ d_X]$
    \end{itemize}
    We say that $\mathit{RG} =  (V, \to)$ is a relevancy graph for $X(v)$ iff $X(v) \in V$. 
\end{definition}
In the remainder of the paper, we use the size of the relevancy graph as a proxy for estimating the effort required to solve a PBES.

\ifnewexample
\begin{example}\label{exa:proof-graph}\label{exa:relevancy-graph}
    A proof graph for the PBES in Example~\ref{exa:core-pbes-running} with initial vertex $X(1)$ is shown in  Fig.~\ref{subfig:running-proofgraph}.
    It shows that vertices $X(1)$, $X(3)$ and $Y(3)$ are $\true$, with the numbers above these vertices indicating their ranks, and the edges showing the required dependencies explaining this solution. The corresponding relevancy graph is shown in  Fig.~\ref{subfig:running-relevancygraph}.   \qedhere
    \begin{figure}[h]
    \centering
        \begin{subfigure}{.48\textwidth}\centering
           \begin{tikzpicture}[scale=0.5,
                every node/.style={scale=0.8},
                node distance = 1cm]
        
                \node[label={[yshift=0.1cm] \small 1}] (X1) {$X(1)$};
                \node[label={[yshift=0.1cm] \small 1},right of=X1,xshift=1.2cm] (X3) {$X(3)$};
                \node[label={[yshift=0.1cm] \small 2},right of=X3,xshift=1.2cm] (Y3) {$Y(3)$};      
                \node[draw=none,below of=X1,yshift=-0.5cm] (bogus) {\phantom{X(2)}};
                
                \path[->]
                (Y3) edge[loop below] node {} (Y3)
                (X1) edge[draw=none,bend right] (X3)
                (X1) edge (X3)
                (X3) edge (Y3)
                ;
            \end{tikzpicture}
            \caption{Proof graph}
            \label{subfig:running-proofgraph}
        \end{subfigure}
        \begin{subfigure}{.5\textwidth}\centering
           \begin{tikzpicture}[scale=0.5,
                every node/.style={scale=0.8},
                node distance = 1cm]
        
                \node (X1) {$X(1)$};
                \node[below of=X1,yshift=-0.5cm] (X2) {$X(2)$};
                \node[right of=X1,xshift=1.2cm] (X3) {$X(3)$};
                \node[left of=X1,xshift=-1.2cm] (Y1) {$Y(1)$};
                \node[left of=X2,xshift=-1.2cm] (Y2) {$Y(2)$};      
                \node[right of=X3,xshift=1.2cm] (Y3) {$Y(3)$};    
                
                \path[->]
                (Y3) edge[loop below] node {} (Y3)
                (X1) edge (Y1)
                (X1) edge[bend right] (X2)
                (X2) edge[bend right] (X1)
                (X1) edge (X3)
                (X2) edge (Y2)
                (X3) edge (Y3)
                ;
            \end{tikzpicture}
            \caption{Relevancy graph}
            \label{subfig:running-relevancygraph}
        \end{subfigure}
        \caption{Proof graph and relevancy graph for the PBES in Example~\ref{exa:core-pbes-running}.}
\end{figure}
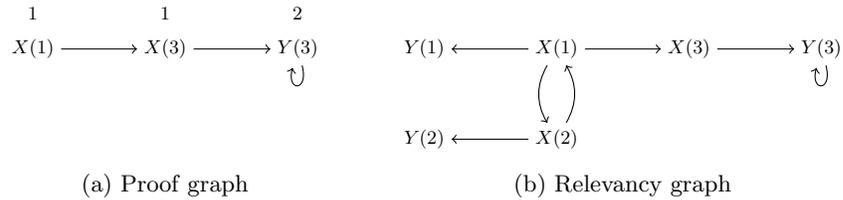
\end{example}
\fi

\subsection{Model Checking and Evidence Generation}
A $\mu$-calculus model checking problem $L(e) \models \varphi$ can be encoded into a PBES using the translation proposed by Wesselink and Willemse~\cite{wesselink_evidence_2018}. A proof graph extracted from a PBES obtained by this encoding allows for generating evidence, in contrast to the encoding of~\cite{groote_model-checking_2005}.
\ifnewexample
\begin{table}[t!]
\caption{Translation to encode, given LPE $L$ and  $\sigma Z .\, \varphi$, the model checking problem $L(e) \models \sigma Z .\, \varphi$ into a PBES~\cite{wesselink_evidence_2018}.}
    \centering
    \begin{tabular}{p{0.2\linewidth} p{0.8\linewidth}}
        \hline
        $\Ec_L(b)$ &= $\epsilon$ \\
        $\Ec_L(Y)$  &= $\epsilon$ \\
        $\Ec_L(\varphi \oplus \psi)$  &= $\Ec_L(\varphi)\,\Ec_L(\psi)$
            \hspace{2cm} \text{for $\oplus \in \{ \lor, \land \}$} \\
        $\Ec_L(\mubox{\alpha}\varphi)$  &= $\Ec_L(\varphi)$ \\
        $\Ec_L(\mudiamond{\alpha}\varphi)$  &= $\Ec_L(\varphi)$ \\
        $\Ec_L(\sigma X .\, \varphi)$  &= $(\sigma X(d_L \colon D_L) = \rhsc_L(\varphi))\,\Ec_L(\varphi)$ \\
        $\rhsc_L(b)$ &= $b$\\
        $\rhsc_L(Y) $&= $Y(d_L) $\\
        $\rhsc_L(\varphi \oplus \psi)$ &= $\rhsc_L(\varphi) \oplus \rhsc_L(\psi) $
            \hspace{2cm} \text{for $\oplus \in \{  \lor, \land \}$} \\
        $\rhsc_L([\alpha]\varphi)$ &=  $\forall_{e_\alpha : E_\alpha} .\, c_\alpha(d,e_\alpha) 
        \implies ((\rhsc_L(\varphi)[g_\alpha(d,e_\alpha)/d] \wedge \Zalphaplus(d,g_\alpha(d,e_\alpha)))$ \\& \hspace{11em} 
        ${} 
        \vee \Zalphaminus(d,g_\alpha(d,e_\alpha)))$\\
        $\rhsc_L(\langle \alpha \rangle \varphi)$ &= $\exists_{e_\alpha : E_\alpha} .\, c_\alpha(d,e_\alpha) 
        \wedge ((\rhsc_L(\varphi)[g_\alpha(d,e_\alpha)/d] \vee \Zalphaminus(d,g_\alpha(d,e_\alpha)))$ \\& \hspace{11em} 
        ${} \wedge \Zalphaplus(d,g_\alpha(d,e_\alpha)))$ \\
        $\rhsc_L(\sigma X.\, \varphi)$ &=  $X(d_L)$\\
        \hline
    \end{tabular}
    \label{tab:table_Ec}
\end{table}
\fi
The translation scheme of the encoding $\Ec_L$ from~\cite{wesselink_evidence_2018} is in Table~\ref{tab:table_Ec}.
\ifnewexample
The predicate variables $\Zalphaplus$ and $\Zalphaminus$ in the right-hand sides $\rhsc_L([\alpha] \varphi)$ and $\rhsc_L(\langle \alpha \rangle \varphi)$ contain information about the action labels, the \textit{transitions}. 
In particular, in a refutation graph, a dependency on $\Zalphaminus$ indicates the $\alpha$-transition is involved in the $\false$ solution, whereas in a proof graph, a dependency on $\Zalphaplus$ indicates the $\alpha$-transition is required for a $\true$ solution.
\fi

\ifnewexample
In addition to the encoding of Table~\ref{tab:table_Ec}, for each action label $\alpha$ equations $\nu \Zalphaplus(d\colon D, d' \colon D)=\true$ and $\mu \Zalphaminus(d\colon D,d ' \colon D) = \false$ are added to the equation system.
These equations are solved and typically grouped at the end of the equation system. We write $\Zplus \subseteq \predvars$ (resp.\ $\Zminus \subseteq \predvars$) for the set of predicate variables $\{ \Zalphaplus \mid \alpha \in \act \}$ (resp.\ $\{ \Zalphaminus \mid \alpha \in \act \}$).
\fi

\ifnewexample
\begin{theorem}[\cite{wesselink_evidence_2018}]\label{thm:PBES}\label{thm:PBES-C} 
    Let $L(e)$ be an LPE, and $\sigma Z.\, \varphi$ be a closed $\mu$-calculus formula. Then, $L(e) \models \sigma Z.\, \varphi$ if and only if $\sem{\Ec_L(\sigma Z.\, \varphi)}(Z)(\sem{e}) = \true$.
\end{theorem}
\fi

We usually write $\pbes$ for the PBES obtained from encoding $\Ec_L$.
\ifnewexample
Specifically, $\pbes = \pbesEc$, where $\pbesEquations$ contains the equations introduced by $\Ec_L$ and $\E_{\Zplus}$ (resp.\ $\E_{\Zminus}$) contains all equations of the shape $\nu \Zalphaplus(d\colon D,d' \colon D) = \true$ (resp.\ $\mu \Zalphaminus(d\colon D,d' \colon D) = \false$). 
\fi

\ifnewexample
\begin{example} \label{exa:running-pbesc}
Recall the LPE and $\mu$-calculus formula from Examples~\ref{exa:running-example} and~\ref{exa:running-formula}.
PBES $\pbes$ consists of the following equations.

    \begin{mcrl2equations}
    (\,    \mu X(s \colon N) &=& \exists_{n: N}. (s \approx 1 \wedge 0 < n < 3 \wedge (X(s+n) \vee \Zaminus(s,s+n)) \wedge \Zaplus(s,s+n) )\\ 
                         &\vee& \exists_{n: N}. (0 < n < s < 3 \wedge (X(s-n) \vee \Zbminus(s,s-n)) \wedge \Zbplus(s,s-n) ) \\
                         &\vee& Y(s) \,)\\
    (\,    \nu Y(s \colon N) &=& s \approx 3 \wedge (Y(s) \vee \Zcminus(s,s)) \wedge \Zcplus(s,s)\,) \\
    (\, \nu \Zaplus(s,s1\colon N)&=& \true \,)\,\,\, 
    (\, \nu \Zbplus(s,s1\colon N)= \true \,)~\,\,  
    (\, \nu \Zcplus(s,s1\colon N)= \true \,) \\ 
    (\, \mu \Zaminus(s,s1\colon N)&=& \false \,)\, 
    (\, \mu \Zbminus(s,s1\colon N)= \false \,)\, 
    (\, \mu \Zcminus(s,s1\colon N)= \false \,)

    \end{mcrl2equations}

The following proof graph for the PBES is found:
\begin{center}
           \begin{tikzpicture}[scale=0.5,
                every node/.style={scale=0.8},
                node distance = 1cm]
        
                \node[label={[yshift=0.1cm] \small 1}] (X1) {$X(1)$};
                \node[label={[yshift=0.1cm] \small 2},left of=X1,xshift=-1cm] (Zaplus) {$\Zaplus(1,3)$};
                \node[label={[yshift=0.1cm] \small 1},right of=X1,xshift=1.7cm] (X3) {$X(3)$};
                \node[label={[yshift=0.1cm] \small 2},right of=X3,xshift=1.7cm] (Y3) {$Y(3)$};      
                \node[label={[yshift=0.1cm] \small 2},right of=Y3,xshift=1cm] (Zcplus) {$\Zcplus(3,3)$};
                
                \path[->]
                (Y3) edge[loop below] node {} (Y3)
                (X1) edge (X3)
                (X3) edge (Y3)
                (X1) edge (Zaplus)
                (Y3) edge (Zcplus)
                ;
            \end{tikzpicture}
\end{center}
Predicate variables $\Zaplus$ and $\Zcplus$ encode information about which $a$-transitions and $c$-transitions in the LPE are involved in proving that the solution to the model checking problem is $\true$.
We filter the relevant vertices with information about evidence from the proof graph, here $\Zaplus(1,3)$ and $\Zcplus(3,3)$, and derive the following LPE (left) and \emph{witness} LTS (right):

\noindent
\parbox[t]{.5\textwidth}{
\begin{mcrl2equations}
    L_w(s: N) &=&
    (s \approx 1) \to  a . L_w(3)\\
            &+& (s \approx 3) \to c . L_w(3)
\end{mcrl2equations}
}
\parbox[t]{.5\textwidth}{
\begin{center}
    \begin{tikzpicture}[node distance=30pt]
        \node [initial left, initial distance=10pt, state] (1) {$1$};
        \node [state, right= of 1] (3) {$3$};
        \draw[->] 
        (3) edge[loop right] node[right] {$c$} (3)
        (1) edge node {$a$} (3)
;
    \end{tikzpicture}
\end{center}
}

We remark that, by construction, this LTS is a subgraph of the LTS in Example~\ref{exa:running-example}, underlying the original specification. 
For the formal definition of witness and counterexample we refer to~\cite{wesselink_evidence_2018}. \qedhere
\end{example}
\fi

\section{Improving Evidence Generation from PBESs}\label{sec:improve-evidence-extraction}\label{sec:evidence-extraction}

The encoding $\Ec_L$ results in a PBES from which evidence supporting the verdict of the model checking problem can be extracted.
However, the additional information added to the right-hand sides of the equations also significantly increases the effort needed to solve the PBES.
We illustrate this using the relevancy graph of the PBES from Example~\ref{exa:running-pbesc}.

\ifnewexample
\begin{example}\label{exa:relavancy-graph-with-counterexample}
The relevancy graph for PBES $\pbes$ from Example~\ref{exa:running-pbesc} is the following.
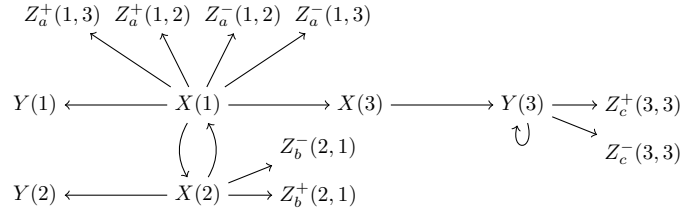
\begin{center}
           \begin{tikzpicture}[scale=0.5,
                every node/.style={scale=0.8},
                node distance = 1.5cm]
        
                \node (X1) {$X(1)$};
                \node[above of=X1, xshift=-0.75cm] (Zaplus1_2) {$\Zaplus(1,2)$};
                \node[above of=X1, xshift=-2.25cm] (Zaplus1_3) {$\Zaplus(1,3)$};
                \node[above of=X1, xshift=0.75cm] (Zaminus1_2) {$\Zaminus(1,2)$};
                
                \node[below of=X1] (X2) {$X(2)$};
                \node[right of=X1,xshift=1.2cm] (X3) {$X(3)$};
                \node[above of=X1,xshift=2.25cm] (Zaminus1_3) {$\Zaminus(1,3)$};
                \node[left of=X1,xshift=-1.2cm] (Y1) {$Y(1)$};
                \node[below of=Y1] (Y2) {$Y(2)$};      
                \node[right of=X2, xshift=0.5cm] (Zbplus2_1) {$\Zbplus(2,1)$};
                \node[above of=Zbplus2_1,yshift=-0.7cm] (Zbminus2_1) {$\Zbminus(2,1)$};
                \node[right of=X3,xshift=1.2cm] (Y3) {$Y(3)$};      
                \node[right of=Y3,xshift=0.5cm] (Zcplus3_3) {$\Zcplus(3,3)$};
                \node[below of=Zcplus3_3,yshift=0.7cm] (Zcminus3_3) {$\Zcminus(3,3)$};
                
                \path[->]
                (Y3) edge[loop below] node {} (Y3)
                (X1) edge (Y1)
                (X1) edge (Zaplus1_2)
                (X1) edge (Zaplus1_3)
                (X1) edge (Zaminus1_2)
                (X1) edge (Zaminus1_3)
                (X2) edge (Zbminus2_1)
                (X2) edge (Zbplus2_1)
                (X1) edge[bend right] (X2)
                (X2) edge[bend right] (X1)
                (X1) edge (X3)
                (X2) edge (Y2)
                (X3) edge (Y3)
                (Y3) edge (Zcminus3_3)
                (Y3) edge (Zcplus3_3)
                ;
            \end{tikzpicture}
\end{center}
Note that it contains dependencies on $\Zaplus(1, n)$ and $\Zaminus(1, n)$ for all $n = 2, 3$. By increasing  the value of $M$, used in the LPE, to values larger than $3$,  the number of vertices can be increased to an arbitrary number. For instance, if $M \approx 1000$ then $X(1)$ will have $1998$ dependencies related to action $a$. The number of dependencies related to action $b$ will increase similarly.\qedhere
\end{example}
\fi

Omitting the information from the PBES that is needed to generate evidence would result in the PBES from Example~\ref{exa:core-pbes-running}, whose much smaller relevancy graph was shown in Fig.~\ref{subfig:running-relevancygraph}.
The relevancy graphs of both PBESs illustrate a trade-off. On the one hand, solving the core PBES $\pbescore$ is (much) more efficient than solving $\pbes$. On the other hand, diagnostic information including the transitions is essential for understanding why a formula is (not) satisfied.

In the remainder of this section, we introduce a three-step approach that allows us to efficiently solve PBESs with additional information for evidence generation. We present the approach assuming that the solution to the PBES is $\true$; the case where the solution is $\false$ is similar.
The three steps are as follows. An overview is presented in the figure on the right.

\parbox{.6\textwidth}{
\begin{enumerate}
  \item Remove the additional information from the PBES $\pbes$, and solve the resulting PBES $\pbescore$ (see Sect.~\ref{sec:solve-core}).
  \item Use the solution of $\pbescore$ to remove unneeded evidence information from the PBES, obtaining $\truepbes$ (see Sect.~\ref{sec:calculate-true}).
  \item Combine the proof graph for $\pbescore$ with $\truepbes$ to obtain a new PBES $\combine{\truepbes, \PG(\pbescore)}$, and solve this PBES (see Sect.~\ref{sec:combine}). 
\end{enumerate}
}
\parbox{.4\textwidth}
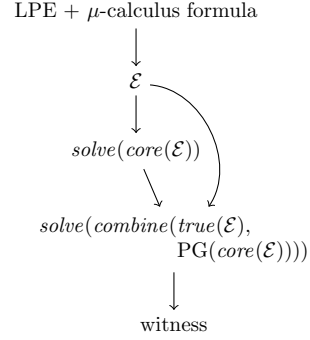
{
\begin{center}
        \begin{tikzpicture}[scale=0.5,
            every node/.style={scale=0.8},
            node distance = 0.5cm]
        
            \node[] (1) {LPE + $\mu$-calculus formula};
            \node[below= of 1] (2) {$\pbes$};
            \node[below= of 2] (3) {$\solve{\pbescore}$};
            \node[below right=0.5cm and -2.4cm of 3] (4) {$\begin{array}{l}\mathit{solve}(\mathit{combine}(\truepbes,\\\phantom{\mathit{solve(combine(}}\PG(\pbescore)))\end{array}$};
            \node[below= of 4] (6) {witness};
             
            \path[->]
            (1) edge node {} (2)
            (2) edge node {} (3)
            (2) edge [bend left=60] node {} (4)
            (3) edge node {} (4)
            (4) edge node {} (6)
            ;
        \end{tikzpicture}
\end{center}
}

The third step results in a solution and proof graph for the original PBES $\pbes$.
In the remainder of this section, we address each of these steps in more detail.

We first introduce some auxiliary notation.
We write $\lambda d_X \colon D_X . \varphi$ lifting predicate formula $\varphi$ to a \emph{predicate function} with the same parameters as predicate variable $X(d_X \colon D_X)$ \cite{orzan_invariants_2010}. 
The semantics is defined as 
$\sem{\lambda d_X \colon D_X . \varphi}\eta \delta =$ $ \lambda v \in \mathbb{D}_X . \sem{\varphi}\eta \delta[v / d_X].$
We use this lifting to substitute a predicate formula for a predicate variable. 
Given predicate formulas $\varphi, \psi$ and predicate variable $X$, we write $\varphi[X := \lambda d_X \colon D_X . \psi]$ to denote that every occurrence of $X$ is replaced with $\lambda d_X \colon D_X . \psi$ in $\varphi$.
We write $\varphi[X := \psi_X, Y := \psi_Y]$ for the simultaneous substitution of $X$ and $Y$ ($X \neq Y$), and generalise this to $\varphi[X := \psi_X]_{X \in \predvars}$ to denote the simultaneous substitution of all $X \in \predvars$.

\subsection{Solving a PBES Without Evidence Information}\label{sec:solve-core}

If we forego evidence, and focus on obtaining a solution for the model checking problem in terms of a \textit{true/false} answer only, the amount of work can be reduced significantly. This motivates the first step in our approach.

A PBES with information about evidence can be simplified by substituting the right-hand sides of solved equations for predicate variables in $\Zplus$ and $\Zminus$. We refer to this as the \emph{core} PBES, defined as follows.
\begin{definition}
Let $\pbes = \pbesEc$. Then
\[
\pbescore = \E_L[\Zaplus := \lambda d \colon D_{\Zaplus} . \true]_{\Zaplus \in \Zplus} [\Zaminus := \lambda d \colon D_{\Zaminus} . \false]_{\Zaminus \in \Zminus} \E_{\Zplus} \E_{\Zminus}
\]
\end{definition}

\begin{example}\label{exa:running-pbes}
    Let PBES $\pbes$ be as in Example~\ref{exa:running-pbesc}. Then $\pbescore$ is obtained from this PBES by replacing the equations for $X$ and $Y$ by the corresponding equations from Example~\ref{exa:core-pbes-running}.
    Note the relevancy graph of the latter (see Fig.~\ref{subfig:running-relevancygraph}) is much smaller than the one for $\pbes$ (see Example~\ref{exa:relavancy-graph-with-counterexample}).\qedhere
\end{example}
It follows immediately from standard results on PBESs~\cite{groote_parameterised_2004} that the solution of the equations in $\pbes_L$ are preserved by transformation $\pbescore$. 
\begin{lemma}
Let $\pbes = \pbesEc$. Then 
$\sem{\pbes} = \sem{\pbescore}$.
\end{lemma}

\subsection{Removing Superfluous Evidence Information}\label{sec:calculate-true}

Once we have established that the solution to $\pbescore$, and hence $\pbes$, is $\true$, only the information for constructing a witness is relevant. We therefore remove the dependencies on predicate variables needed to construct counterexamples.

\begin{definition}
Let $\pbes = \pbesEc$. Then
\[
\truepbes = \pbesEquations[\Zaminus := \lambda d \colon D_{\Zaminus} . \false]_{\Zaminus \in \Zminus}\E_{\Zplus}\E_{\Zminus}
\]
\end{definition}
By definition of $\pbescore$ and $\truepbes$, the following result follows immediately from the semantics~\cite{groote_parameterised_2004}.
\begin{lemma}
Let $\pbes = \pbesEc$, then 
$\sem{\pbescore} = \sem{\truepbes}$.
\end{lemma}

\ifnewexample
\begin{example}\label{exa:running-true-pbes}
    Recall that the solution of $\pbescore$ of Example~\ref{exa:running-pbes} is $\true$. We use this to obtain the following PBES $\truepbes$:
    \begin{mcrl2equations}
    (\,    \mu X(s \colon N) &=& (\exists_{n: N}. (s \approx 1 \wedge 0 < n < 3 \wedge X(s+n) \wedge \Zaplus(s,s+n) ))\\ 
                         &\vee& (\exists_{n: N}. (0 < n < s < 3 \wedge X(s-n) \wedge \Zbplus(s,s-n) )) \\
                         &\vee& Y(s)\,)\\
    (\,    \nu Y(s \colon N) &=& s \approx 3 \wedge Y(s) \wedge \Zcplus(s,s) \,)\\
    (\, \nu \Zaplus(s,s1\colon N)&=& \true \,)\,\,\, 
    (\, \nu \Zbplus(s,s1\colon N)= \true \,)~\,\,  
    (\, \nu \Zcplus(s,s1\colon N)= \true \,) \\ 
    (\, \mu \Zaminus(s,s1\colon N)&=& \false \,)\,
    (\, \mu \Zbminus(s,s1\colon N)= \false \,)~\, 
    (\, \mu \Zcminus(s,s1\colon N)= \false \,)

    \end{mcrl2equations}
    The corresponding relevancy graph is obtained by removing all vertices for $\Zaminus, \Zbminus$ and $\Zcminus$ and their incoming edges from the relevancy graph in Example~\ref{exa:relavancy-graph-with-counterexample}.\qedhere
\end{example}
\fi

\subsection{Simplifying a PBES using Evidence Information}\label{sec:combine}

We now show how a proof graph can be used to further simplify the right-hand sides in a PBES.
For this, recall that for a vertex $X(v)$ in the proof graph, the successors of $X(v)$ yield a predicate environment that makes $\varphi_X[d_X := v]$ true.
Using this information, we can syntactically remove all dependencies that are not in the proof graph from the right-hand sides in a PBES, without affecting the solution.
To achieve this, we define $\combine{\pbes, \pgG}$ as follows.

\begin{definition}\label{def:combine}
Let $\pbes$ be a PBES, and $\pgG = (V, E)$ be a proof graph. 
Then PBES $\combine{\pbes, \pgG}$ is obtained by replacing the right-hand side of every equation $\sigma X(d_X: D_X) = \varphi_X$ in $\pbes$ by the formula
\[
    \varphi_X[Y := \lambda e . \bigwedge_{v \in V_X} (d_X \approx v \implies e \in E_{X(v),Y} \land Y(e))]_{Y \in \predvars \setminus (\Zplus \cup \Zminus)}
\]
where $V_X = \{ v \in \mathbb{D} \mid X(v) \in V \}$ contains all values $v$ such that $X(v)$ is a vertex in $\pgG$, and $E_{X(v),Y} = \{ w \in \mathbb{D} \mid \langle X(v), Y(w) \rangle \in E \}$ contains all direct dependencies of $X(v)$ on $Y$ in the proof graph.
\end{definition}
Intuitively, this retains only those dependencies in $\varphi_X$ that according to the proof graph are needed to show that $X(v)$ is true, for any $v$.

\ifnewexample
\begin{example}
    Recall the equation for $X$ in PBES $\truepbes$ from Example~\ref{exa:running-true-pbes}.
    \begin{mcrl2equations}
    (\,    \mu X(s \colon N) &=& (\exists_{n: N}. (s \approx 1 \wedge 0 < n < 3 \wedge X(s+n) \wedge \Zaplus(s,s+n) ))\\ 
                         &\vee& (\exists_{n: N}. (0 < n < s < 3 \wedge X(s-n) \wedge \Zbplus(s,s-n) )) \\
                         &\vee& Y(s) \,)
    \end{mcrl2equations}
    The proof graph for $\pbescore$ (see Fig.~\ref{subfig:running-proofgraph}) has vertices $V = \{ X(1), X(3), Y(3) \}$ and edges $E = \{ (X(1), X(3)), (X(3), Y(3)), (Y(3), Y(3)) \}$.
    So, we infer $V_X = \{ 1, 3 \}$, $V_Y = \{ 3 \}$, $E_{X(1),X} = \{ 3 \}$, $E_{X(1), Y} = E_{X(3), X} = \emptyset$, $E_{X(3), Y} = \{ 3 \}$, $E_{Y(3), X} = \emptyset$, and $E_{Y(3), Y} = \{ 3 \}$.

    The right-hand side of the equation for $X$ in $\combine{\truepbes, \PG(\pbescore)}$, after $\beta$-reduction and simplification, is as follows.
    \begin{mcrl2equations}
    (\,    \mu X(s \colon N) &=& (\exists_{n: N}. (s \approx 1 \wedge 0 < n < 3 \\
        & & {} \wedge ( s \approx 1 \implies (s+n) \in \{3\} \wedge X(s+n) )  \\ 
        & & {} \wedge ( s \approx 3 \implies (s+n) \in \emptyset \wedge X(s+n) )  \wedge \Zaplus(s,s+n) ))\\ 
        &\vee& (\exists_{n: N}. (0 < n < s < 3 \\
        & & {} \wedge (s \approx 1 \implies (s-n) \in \{3\} \wedge X(s-n) ) \\
        & & {} \wedge (s \approx 3 \implies (s-n) \in \emptyset \wedge X(s-n) ) \wedge \Zbplus(s,s-n) )) \\
        &\vee& (s \approx 3 \implies s \in \{3\} \wedge Y(s)) \,)
    \end{mcrl2equations}
    This simplifies further to
    \begin{mcrl2equations}
    (\,    \mu X(s \colon N) &=& (\exists_{n: N}. (s \approx 1 \wedge n \approx 2 \wedge X(s+n)  \wedge \Zaplus(s,s+n)
         ) )\\ 
        &\vee& (\exists_{n: N}. (0 < n < s < 3 \wedge \Zbplus(s,s-n) )) \\
        &\vee& (s \approx 3 \implies Y(s)) \,)
    \end{mcrl2equations}

    If we also apply the corresponding substitution to the equation for $Y$ and apply some simplification, we obtain the following PBES.
    
    \begin{mcrl2equations}
    (\,    \mu X(s \colon N) &=& (\exists_{n: N}. (s \approx 1 \wedge n \approx 2 \wedge X(s+n)  \wedge \Zaplus(s,s+n)
         ) )\\ 
        &\vee& (\exists_{n: N}. (0 < n < s < 3 \wedge \Zbplus(s,s-n) )) \\
        &\vee& (s \approx 3 \implies Y(s)) \,)\\
    (\,    \nu Y(s \colon N) &=& s \approx 3 \wedge
            Y(s) \wedge \Zcplus(s, s) \,)\\
    (\, \nu \Zaplus(s,s1\colon N)&=& \true \,)\,\,\, 
    (\, \nu \Zbplus(s,s1\colon N)= \true \,)~\,\,  
    (\, \nu \Zcplus(s,s1\colon N)= \true \,) \\ 
    (\, \mu \Zaminus(s,s1\colon N)&=& \false \,)\, 
    (\, \mu \Zbminus(s,s1\colon N)= \false \,)~\, 
    (\, \mu \Zcminus(s,s1\colon N)= \false \,)
    \end{mcrl2equations}

    This PBES has the following relevancy graph, that no longer has dependencies on $\Zbplus$ and only a single dependency on $\Zaplus$, and is significantly smaller than the relevancy graph of Example~\ref{exa:relavancy-graph-with-counterexample}: 
    \begin{center}
          \begin{tikzpicture}[scale=0.5,
                every node/.style={scale=0.8},
                node distance = 1.5cm]
        
                \node (X1) {$X(1)$};
                \node[left of=X1, xshift=-0.5cm] (Zaplus1_3) {$\Zaplus(1,3)$};
                \node[right of=X1,xshift=1.2cm] (X3) {$X(3)$};
                \node[right of=X3,xshift=1.2cm] (Y3) {$Y(3)$};      
                \node[right of=Y3, xshift=0.5cm] (Zcplus3_3) {$\Zcplus(3,3)$};
                
                \path[->]
                (Y3) edge[loop below] node {} (Y3)
                (X1) edge (Zaplus1_3)
                (X1) edge (X3)
                (X3) edge (Y3)
                (Y3) edge (Zcplus3_3)
                ;
            \end{tikzpicture} \qedhere   \end{center}   
\end{example}
\fi

In the example we have combined the PBES $\truepbes$ with a proof graph for the strongly related PBES $\pbescore$. 
Towards establishing the correctness of this transformation, we first prove the following technical lemma, that shows that the substitution of a single predicate variable in a proof graph like context does not change the solution.

\begin{lemma}\label{lem:proof-graph-substitution-preserves-solution}
Let $V \subseteq \mathbb{D}$ be a set of values, $Y$ a predicate variable, and $\{ E_{v,Y} \subseteq \mathbb{D} \}_{v \in V}$ be a $V$-indexed family of sets of values.
For every $v \in V$, predicate environment $\eta$ such that $\eta(Y)(w) = \true$ iff $w \in E_{v,Y}$, predicate formula $\varphi$ of data variable $d$, and data environment $\delta$, we have
\[
\sem{\varphi}\eta\delta[v/d] \implies \sem{\varphi[Y := \lambda e . \bigwedge_{w \in V}( d \approx w \implies e \in E_{v, Y} \land Y(e)]}\eta\delta[v/d]
\]
\end{lemma}
\begin{proof}
Fix, $V$, $Y$, $\{ E_{v,Y} \subseteq \mathbb{D} \}_{v \in V}$, $v$ and $\eta$ as in the statement of the lemma.
We proceed by induction on the structure of $\varphi$.
Most cases are immediate, or follow from the induction hypothesis and the semantics of predicate formulas.
We focus on the interesting case where $\varphi = Z(e')$ for some $Z$ and $e'$.
If $Z \neq Y$, the result is immediate, since the substitution has no effect.
So, suppose $Z = Y$.
We have to show that
$\sem{Y(e')}\eta\delta[v/d]$ implies $\sem{Y(e')[Y := \lambda e . \bigwedge_{w \in V} (d \approx w \implies e \in E_{v,Y} \land Y(e))]}\eta\delta[v / d]$.

Assume $\sem{Y(e')}\eta\delta[v/d]$ is $\true$.
Hence, $\eta(Y)(\sem{e'}\delta[v/d])$ is $\true$, so by assumption, $\sem{e'}\delta[v/d] \in E_{v,Y}$, and therefore $\sem{e'\in E_{v,Y}}\delta[v/d]$ is $\true$.
Similarly, it immediately follows that $\sem{d \approx w}\delta[v/d]$ is $\true$ iff $w = v$.
So, it follows that
$\sem{\bigwedge_{w \in V} (d \approx w \implies e' \in E_{v,Y} \land Y(e'))}\eta\delta[v / d]$ is $\true$.
Using the definition of substitution and $\beta$-reduction, it then follows that
$\sem{Y(e')[Y := \lambda e . \bigwedge_{w \in V} (d \approx w \implies e \in E_{v,Y} \land Y(e))]}\eta\delta[v / d]$ is also $\true$. \qedhere
\end{proof}

We use this lemma to establish that the proof graph for $\pbescore$ can easily be extended into a proof graph for $\combine{\truepbes, G}$.
This shows that for values of interest to the original model checking problem, the result remains unchanged.

\begin{proposition}\label{prop:PG-core-is-PG-combine}
For every proof graph $G$ for $\pbescore$, there is a proof graph $G'$ for $\combine{\truepbes, G}$ such that $G$ is a subgraph of $G'$.
\end{proposition}
\begin{proof}
Let $G = (V, E)$ be a proof graph for $\pbescore$. Define $G' = (V \cup V_{\Zplus}, E \cup (V \times V_{\Zplus}))$ with $V_{\Zplus} = \{ \Zaplus(e_L, e_a, e'_L) \mid \Zaplus \in \Zplus, e_L, e'_L \in \mathbb{D}_L, e_a \in \mathbb{D}_a \}$.
Clearly $G$ is a subgraph of $G'$.

We show that $G'$ is a proof graph for $\combine{\truepbes, G}$.
Note that vertices in $V_{\Zplus}$ do not appear on any infinite path, so the infinite paths in the graph are not changed.
For every $X$ bound in $\pbescore$, let $\varphi_X$ be the right-hand side of $X$ in $\pbescore$, and let 
\[
\psi_X := \varphi_X[Y := \lambda e . \bigwedge_{v \in V_X} (d_X \approx v \implies e \in E_{X(v),Y} \land Y(e))]_{Y \in \predvars \setminus (\Zplus \cup \Zminus)}
\]
be the right-hand side of $X$ in $\combine{\truepbes, G}$.
We need to prove that for every $X(v) \in V \cup V_{\Zplus}$ and for all $\delta$, 
$ 
    \sem{\psi_X}\eta_{X(v)} \delta[v/d_X] 
$
is $\true$, where $V_X = \{ v \in \mathbb{D} \mid X(v) \in V \}$ and
$E_{X(v),Y} = \{ w \in \mathbb{D} \mid \langle X(v), Y(w) \rangle \in E \}$
according to Definition~\ref{def:combine}, and $\eta_{X(v)}$ is such that $\eta_{X(v)}(Y)(w) = \true$ iff $\langle X(v), Y(w) \rangle \in E$ according to Definition~\ref{def:proof_graph}.
For $X(v) \in V_{\Zplus}$, the result follows immediately, as $\psi_X = \true$.
So, suppose $X(v) \in V$.
As $G$ is a proof graph for $\pbescore$, $\sem{\varphi_X}\eta_{X(v)} \delta[v/d_X]$ is $\true$.
Note $V_X$, $E_{X(v),Y}$ (for $Y \in \predvars \setminus (\Zplus \cup \Zminus)$) and $\eta_{X(v)}$ satisfy the conditions of Lemma~\ref{lem:proof-graph-substitution-preserves-solution}, so using the definition of $\psi_X$ and repeated application of the lemma for $Y \in \predvars \setminus (\Zplus \cup \Zminus)$,
it follows that 
$
    \sem{\psi_X}\eta_{X(v)} \delta[v/d_X]
$.
Hence $G'$ is a proof graph for $\combine{\truepbes, G}$. \qedhere
\end{proof}

\subsection{Providing Evidence for the Original PBES}\label{sec:correctness}

The result that every proof graph $G$ for $\pbescore$ can be extended into a proof graph $G'$ for $\combine{\truepbes, G}$ is sufficient to show that $\combine{\truepbes, G}$ does not change the solution of the PBES. However, $G'$ may contain too many variables with evidence information, resulting in witnesses that are larger than needed.
In practice, we therefore solve $\combine{\truepbes, G}$ again, leading to a proof graph that only contains the necessary dependencies on variables in $\Zplus$.

Correctness of our approach ultimately follows if a solution and proof graph obtained for PBES $\combine{\truepbes, G}$ are also a correct solution and proof graph for our original PBES $\pbes$.
We first establish the following result.

\begin{lemma}\label{lem:proof-graph-substitution-preserves-solution-reverse}
Let $V \subseteq \mathbb{D}$ be a set of values, $Y$ a predicate variable, and $\{ E_{v,Y} \subseteq \mathbb{D} \}_{v \in V}$ be a $V$-indexed family of sets of values.
For every $v \in V$, predicate environment $\eta$,
formula $\varphi$ over data variable $d$, and data environment $\delta$,
\[
\sem{\varphi[Y := \lambda e . \bigwedge_{w \in V}( d \approx w \implies e \in E_{v, Y} \land Y(e)]}\eta\delta[v/d] \implies \sem{\varphi}\eta\delta[v/d]
\]
\end{lemma}
\begin{proof}
Fix, $V$, $Y$, $\{ E_{v,Y} \subseteq \mathbb{D} \}_{v \in V}$, $v$ and $\eta$ as in the statement of the lemma.
We proceed by induction on the structure of $\varphi$.
Most cases are immediate, or follow from the induction hypothesis and the semantics of predicate formulas.
We focus on the interesting case where $\varphi = Z(e')$ for some $Z$ and $e'$.
If $Z \neq Y$, the result is immediate, since the substitution has no effect.
So, suppose $Z = Y$.

Assume $\sem{Y(e')[Y := \lambda e . \bigwedge_{w \in V} (d \approx w \implies e \in E_{v,Y} \land Y(e))]}\eta\delta[v / d]$ is $\true$.
Using the definition of substitution and $\beta$-reduction, $\sem{\bigwedge_{w \in V} (d \approx w \implies e' \in E_{v,Y} \land Y(e'))}\eta\delta[v / d]$ is $\true$.
This implies that for every $w \in V$, $\sem{(d \approx w \implies e' \in E_{v,Y} \land Y(e'))}\eta\delta[v / d]$ is $\true$. Since $v \in V$, in particular
 $\sem{(d \approx v \implies e' \in E_{v,Y} \land Y(e'))}\eta\delta[v / d]$ is $\true$.
So, according to the semantics,
if $\sem{d \approx v}\eta\delta[v / d]$ is $\true$ then $\sem{e' \in E_{v,Y} \land Y(e')}\eta\delta[v / d]$ is also $\true$.
That $\sem{d \approx v}\eta\delta[v / d]$ follows directly from the semantics, so 
$\sem{e' \in E_{v,Y} \land Y(e')}\eta\delta[v / d]$ is also $\true$, hence in particular $\sem{Y(e')}\eta\delta[v / d]$ is $\true$. \qedhere
\end{proof}
The lemma establishes that the proof graph we compute for $\combine{\truepbes, G}$ is also a proof graph for $\pbes$.

\begin{theorem}\label{thm:proof-graph-combine-is-proof-graph-E}
Let $G$ be a proof graph for $\pbescore$.
Then every proof graph $G'$ for $\combine{\truepbes, G}$ is also a proof graph for $\pbes$.
\end{theorem}
\begin{proof}
Let $G = (V, E)$ be a proof graph for $\pbescore$ and $G' = (V', E')$ be a proof graph for $\combine{\truepbes, G}$.
For every $X$ bound in $\pbescore$, let $\varphi_X$ be the right-hand side of $X$ in $\pbes$. Note that $\varphi^t_X := \varphi_X[\Zaminus := \lambda d \colon D_{\Zaminus} . \false]_{\Zaminus \in \Zminus}$ is the right-hand side of $X$ in $\truepbes$.
Let 
\[
\psi_X := \varphi^t_X[Y := \lambda e . \bigwedge_{v \in V_X} (d_X \approx v \implies e \in E_{X(v),Y} \land Y(e))]_{Y \in \predvars \setminus (\Zplus \cup \Zminus)}
\]
where $V_X = \{ v \in \mathbb{D} \mid X(v) \in V \}$ and
$E_{X(v),Y} = \{ w \in \mathbb{D} \mid \langle X(v), Y(w) \rangle \in E \}$, 
be the right-hand side of $X$ in $\combine{\truepbes, G}$.

We prove that $G'$ is a proof graph for $\truepbes$.
As all variables in $\Zminus$ are $\false$, no proof graph for $\pbes$ requires dependencies on these variables. As $\truepbes$ only removes these variables from the right-hand sides, if $G'$ is a proof graph for $\truepbes$ it immediately is a proof graph for $\pbes$.

To prove that $G'$ is a proof graph for $\truepbes$, since it already is a proof graph for $\combine{\truepbes, G}$, it suffices to show that for every $X(v) \in V'$ and for all $\delta$,
$
\sem{\varphi^t_X}\eta_{X(v)}\delta[v/d_X]
$
is $\true$, where $\eta_{X(v)}$ is such that $\eta_{X(v)}(Y)(w) = \true$ iff $\langle X(v), Y(w) \rangle \in E'$.

Fix $X(v) \in V'$. As $G'$ is a proof graph for $\combine{\truepbes, G}$, we know that
$ 
    \sem{\psi_X}\eta_{X(v)} \delta[v/d_X] 
$ is $\true$.

Note that $V_X = \{ v \in \mathbb{D} \mid X(v) \in V \}$ and
$E_{X(v),Y} = \{ w \in \mathbb{D} \mid \langle X(v), Y(w) \rangle \in E \}$ (for $Y \in \predvars \setminus (\Zplus \cup \Zminus)$), as used in the definition of $\psi_X$ satisfy the conditions of Lemma~\ref{lem:proof-graph-substitution-preserves-solution-reverse}, so using repeated application of the lemma for $Y \in \predvars \setminus (\Zplus \cup \Zminus)$,
it follows that 
$
\sem{\varphi^t_X}\eta_{X(v)}\delta[v/d_X]
$.
Hence $G'$ is a proof graph for $\truepbes$, thus also for $\pbes$.
\qedhere
\end{proof}

Hence, the proof graph computed using our approach is a proof graph for the original PBES $\pbes$, and the witness we extract from it is a witness for the model checking problem encoded by $\pbes$.

\section{Implementation and Evaluation}\label{sec:implementation-evaluation}
The mCRL2 toolset~\cite{BGK+2019} supports the original approach to extract evidence from PBESs~\cite{wesselink_evidence_2018} in the explicit model checking tool \emph{pbessolve}.
We have extended this tool with the approach described in Sect.~\ref{sec:improve-evidence-extraction}.
The implementation does not precompute $\combine{\truepbes, \PG(\pbescore)}$. Instead, the corresponding right-hand sides are computed on-the-fly.

We have also extended the tool \emph{pbessolvesymbolic}, that supports symbolic solving of PBESs~\cite{laveaux_--fly_2022}, with a hybrid approach that enables evidence generation for symbolic model checking.
In this approach, $\pbescore$ is represented and solved symbolically. This results in a symbolic characterisation of $\PG(\pbescore)$. To obtain this proof graph, the symbolic implementation of Zielonka's algorithm has been extended in such a way that an over-approximation of the proof graph is efficiently computed.\footnote{The over-approximation is constructed such that it again is a proof graph, and has an edge-relation that has a compact symbolic representation.}
In order to reason symbolically about the underlying proof graph, the PBES must be in standard recursive form (SRF)~\cite{neele-thesis}, namely, every right-hand side is either disjunctive or conjunctive. This is not a restriction, as any PBES can be transformed into this format.
Exploring and solving $\combine{\truepbes, \PG(\pbescore)}$ is done explicitly. The implementation here uses the (symbolic) proof graph to again compute right-hand sides on-the-fly.

In both cases, the resulting parity game is solved using an explicit version of Zielonka's algorithm that results in \emph{minimal} proof graph~\cite{wesselink_evidence_2018}.

\subsection{Experimental Setup}
We evaluate the effectiveness of our approach using a number of mCRL2 specifications with $\mu$-calculus formulas.
Each mCRL2 specification is linearised into an LPE, and combined with a $\mu$-calculus formula into a PBES encoding the corresponding model checking problem.
To evaluate the effect of our improvements we compare the six different approaches to solve PBESs that are available in the mCRL2 toolset.
For explicit model checking, we compare directly solving the PBES with information about evidence~\cite{wesselink_evidence_2018} (\texttt{n-expl}), and our own approach (\texttt{expl}).
For symbolic model checking, the comparison is similar, but we use the symbolic algorithms from~\cite{laveaux_--fly_2022} to directly solve the PBESs with evidence (\texttt{n-symb}). We compare it to the hybrid implementation of our approach (\texttt{symb}).
To illustrate the overhead of solving PBES with information about evidence, we also include directly solving the PBES without that information explicitly~\cite{BGK+2019,groote_model-checking_2005} (\texttt{noCE-expl}), and symbolically (\texttt{noCE-symb})~\cite{laveaux_--fly_2022}.

The experiments are run using different types of models.
This includes our running example scaled to $M=1000$ (\textbf{witness1000}). 
We also use models based on industrial applications: the Storage Management System (\textbf{SMS}) and the Workload Management System (\textbf{WMS}) of the DIRAC Community Grid Solution for the LHCb experiment at CERN~\cite{remenska_using_2013}; the IEEE 1394 (\textbf{1394-fin}) interface standard that specifies a serial bus architecture for high-speed communications~\cite{garavel_four_2024}; two versions of the ERTMS Hybrid Level 3 train control system specification each with a different implementation of the Trackside System~\cite{bartholomeus_modelling_2018}, \textit{immediate update} (\textbf{ertms-hl3}) and \textit{simultaneous update} (\textbf{ertms-hl3su}); and a Mechanical Lung Ventilator~\cite{dortmont-MLV-2024} (\textbf{MLV}).
Moreover, we include a model of the onebit sliding window protocol (\textbf{onebit}) with buffers of size 2; and a model of the Hesselink's handshake register~\cite{hesselink_invariants_1998} (\textbf{hesselink}).
For each of these models we verify requirements that are described in the corresponding papers.
We include model checking problems that hold (\cmark), and ones that do not hold (\xmark).

All experiments are run 10 times, on a machine with 4 Intel 6136 CPUs and 3TB of RAM, running Ubuntu 20.04. We used a time-out of 1 hour (3600 seconds), and a memory limit of 64GB.
For models \textbf{ertms-hl3}, \textbf{WMS} and \textbf{MLV} only the cases \texttt{noCE-symb} and \texttt{symbolic} were run 10 times. A preliminary experiment showed that all other cases either time-out or run out-of-memory.
A reproduction package is available from \url{https://doi.org/10.5281/zenodo.14616612}.

\begin{table}
\caption{Experimental results for model checking, reporting number of vertices in the relevancy graph, and the mean total time over 10 runs (highlighted). For \texttt{expl} and \texttt{symb} we report the number of vertices in the relevancy graph after the second solving; the first solving results in the numbers reported in \texttt{noCE-expl} and \texttt{noCE-symb}, respectively. For every case, the fastest a) of \texttt{n-expl} and \texttt{expl}, and b) of \texttt{n-symb} and \texttt{symb} are highlighted. }
\label{tab:experiments}
\begin{center}
  \renewcommand{\arraystretch}{1}
    \scriptsize
    \begin{tabular}{r c r r r r r r}& Result & \texttt{noCE-expl}  & \texttt{n-expl} & \texttt{expl} & \texttt{noCE-symb}  & \texttt{n-symb} & \texttt{symb}\\
      \toprule
      \textbf{witness1000}& \multicolumn{7}{c}{}\\
 \emph{canDobAlways}&\cmark& 2\,000& 1\,001\,002& 5& 2\,000& --& 5 \\
 \rowcolor[gray]{0.95}& & 74.5s & 170.3s & \textbf{74.6}s & 251.0s & t-o & \textbf{247.1}s \\

      \toprule
      \textbf{SMS}& \multicolumn{7}{c}{}\\
 \emph{eventuallyDeleted}&\xmark& 25\,206& 195\,406& 1\,503& 27\,506& --& 2\,443 \\
 \rowcolor[gray]{0.95}& & 0.9s & 4.7s & \textbf{1.1}s & 1.7s & o-o-m & \textbf{2.1}s \\

 \emph{noTransitFromDeleted}&\xmark& 16\,106& 187\,338& 28& 18\,886& 504\,726& 68 \\
 \rowcolor[gray]{0.95}& & 0.6s & 3.6s & \textbf{0.8}s & 1.7s & 118.1s & \textbf{2.4}s \\

      \toprule
      \textbf{hesselink}& \multicolumn{7}{c}{}\\
 \emph{valuesCanBeRead}&\cmark& 1\,093\,760& 3\,325\,184& 2\,209\,472& 1\,093\,760& --& 2\,209\,472 \\
 \rowcolor[gray]{0.95}& & 36.3s & \textbf{135.8}s & 143.6s & 21.5s & t-o & \textbf{133.0}s \\

      \toprule
      \textbf{1394-fin}& \multicolumn{7}{c}{}\\
 \emph{noDeadlockUpgrade}&\cmark& 377\,138& 1\,034\,224& 705\,681& 377\,138& --& 705\,681 \\
 \rowcolor[gray]{0.95}& & 144.4s & \textbf{277.6}s & 405.9s & 26.1s & t-o & \textbf{310.9}s \\

 \emph{noDoubleConfirmation}&\cmark& 565\,708& 1\,222\,794& 894\,251& 565\,708& --& 894\,251 \\
 \rowcolor[gray]{0.95}& & 190.6s & \textbf{235.8}s & 405.9s & 9.0s & t-o & \textbf{240.8}s \\

 \emph{noDeadlock}&\cmark& 188\,569& 845\,655& 517\,112& 188\,569& --& 517\,112 \\
 \rowcolor[gray]{0.95}& & 81.2s & \textbf{209.7}s & 279.4s & 27.2s & t-o & \textbf{245.2}s \\

      \toprule
      \textbf{onebit}& \multicolumn{7}{c}{}\\
 \emph{messCanBeOvertaken}&\xmark& 164\,352& 1\,100\,672& 632\,512& 164\,352& --& 632\,512 \\
 \rowcolor[gray]{0.95}& & 7.1s & 39.3s & \textbf{37.2}s & 4.5s & o-o-m & \textbf{34.6}s \\

 \emph{messReadInevSent}&\xmark& 153\,984& 1\,090\,304& 4& 153\,984& --& 112\,981 \\
 \rowcolor[gray]{0.95}& & 6.5s & 31.2s & \textbf{6.8}s & 3.5s & t-o & \textbf{8.0}s \\

 \emph{noDeadlock}&\cmark& 81\,920& 1\,018\,240& 550\,080& 81\,920& --& 550\,080 \\
 \rowcolor[gray]{0.95}& & 3.4s & 31.7s & \textbf{29.0}s & 2.3s & o-o-m & \textbf{27.9}s \\

      \toprule
      \textbf{ertms-hl3su}& \multicolumn{7}{c}{}\\
 \emph{detStabilisation}&\cmark& --& --& --& 11\,973\,823& --& -- \\
 \rowcolor[gray]{0.95}& & t-o & t-o & t-o & 378.9s & t-o & o-o-m \\

 \emph{termination}&\xmark& 188\,865& --& 13& 196\,593& --& 29 \\
 \rowcolor[gray]{0.95}& & 3\,083.3s & t-o & \textbf{3\,087.2}s & 342.4s & t-o & \textbf{352.1}s \\

      \toprule
      \textbf{ertms-hl3}& \multicolumn{7}{c}{}\\
 \emph{termination}&\xmark& --& --& --& 321\,421& --& 90 \\
 \rowcolor[gray]{0.95}& & t-o & t-o & t-o & 511.9s & t-o & \textbf{514.1}s \\

 \emph{detStabilisation}&\xmark& --& --& --& 17\,756\,789& --& 685 \\
 \rowcolor[gray]{0.95}& & t-o & t-o & t-o & 364.9s & t-o & \textbf{406.0}s \\

      \toprule
      \textbf{MLV}& \multicolumn{7}{c}{}\\
 \emph{scenarioResumeVentilation}&\cmark& --& --& --& 6.15131e+23& --& 5\,950 \\
 \rowcolor[gray]{0.95}& & t-o & t-o & t-o & 1\,663.7s & t-o & \textbf{1\,827.0}s \\

 \emph{CONT38}&\xmark& --& --& --& 5.08225e+23& --& 5\,968 \\
 \rowcolor[gray]{0.95}& & t-o & t-o & t-o & 1\,519.8s & t-o & \textbf{1\,565.0}s \\

      \toprule
      \textbf{WMS}& \multicolumn{7}{c}{}\\
 \emph{jobFailedToDone}&\xmark& --& --& --& 269\,767\,184& --& 226 \\
 \rowcolor[gray]{0.95}& & t-o & o-o-m & t-o & 20.8s & t-o & \textbf{28.1}s \\

 \emph{noZombieJobs}&\xmark& --& --& --& 316\,631\,360& --& 38 \\
 \rowcolor[gray]{0.95}& & t-o & o-o-m & t-o & 24.7s & t-o & \textbf{56.2}s \\

    \bottomrule
\end{tabular}
\end{center}
\end{table}

\subsection{Results and Discussion}
The results are presented in Table~\ref{tab:experiments}. We highlight the fastest run with counterexample information for both the explicit and symbolic cases.
We report the number of vertices in the relevancy graph generated to solve the model checking problem and the mean total running time of ten runs in seconds (`t-o' for time-out, `o-o-m' for out-of-memory).
The standard deviation is typically below
10\% of the mean.\footnote{The SDs for the only cases where it exceeds 10\% of the mean are: case \texttt{noCE-expl} \textbf{SMS} \textit{eventuallyDeleted} and \textit{noTransitFromDeleted}: 0.1; case \texttt{noCE-symb} \textbf{hesselink}: 3.2, \textbf{1394-fin} \textit{noDoubleConfirmation}: 5.3 and \textit{noDeadlock}: 4.9, \textbf{WMS} \textit{noZombieJobs}: 3.0; and case \texttt{n-symb} \textbf{WMS} \textit{jobFailedToDone}: 6.3 and \textit{noZombieJobs}: 8.0} 

We focus our discussion on the approaches that support evidence generation.
For the explicit implementation, our approach (\texttt{expl}) is typically comparable with the original approach (\texttt{n-expl}).
In some cases, we see that our approach reduces the running time of the verification in comparison with the original one, e.g., for model \textbf{onebit} and property \textit{messReadInevSent}. This is typically the case when the evidence is small, and in these cases the running time is similar to that of solving the PBES without additional information (\texttt{noCE-expl}).
In other cases we instead see that \texttt{expl} has some overhead, e.g., for model \textbf{1394-fin} and property \textit{noDoubleConfirmation}.
Closer inspection suggests the evidence in these cases comprises most of the state space.
Since our \texttt{expl} approach is a two-step approach, essentially the full exploration is performed twice,  resulting in a larger running time.

Moreover, the experiments show that our approach for evidence generation in the context of symbolic model checking (\texttt{symb}) always outperforms the original approach (\texttt{n-symb}), typically enabling evidence generation for symbolic model checking, while this is infeasible with the original approach.

\section{Conclusion}\label{sec:conclusion}

In this paper we have described an approach to solving PBESs that allows for efficient evidence generation.
Our approach solves a PBES without evidence information and uses its solution to simplify solving the PBES with evidence information as described in~\cite{wesselink_evidence_2018}.
We have established correctness of our approach, and implemented this in the mCRL2 toolset as part of an explicit and a symbolic model checking tool.

Our evaluation shows that for explicit model checking, the performance is comparable to the original approach to evidence generation from~\cite{wesselink_evidence_2018}. In case the counterexample is small, little overhead is incurred compared to solving the PBES without evidence information.
Our approach makes evidence generation from PBESs efficient for symbolic checking, whereas this was not feasible before.

We plan to integrate our approach with other optimisations in the PBES solvers in the mCRL2 toolset, and to preserve evidence information in static analysis techniques that are often used as preprocessing~\cite{keiren_improved_2013,orzan_invariants_2010}.

\subsubsection*{Acknowledgements} This work was supported by the MACHINAIDE project (ITEA3, No. 18030), the National Growth Fund through the Dutch 6G flagship project “Future Network Services”, and the Cynergy4MIE project (ChipsJU, No. 101140226).

\bibliographystyle{splncs04}
\bibliography{bibfile}

\end{document}